\documentclass[conference]{IEEEtran}
\IEEEoverridecommandlockouts
\usepackage{cite}
\usepackage{amsmath,amssymb,amsfonts}
\usepackage{algorithmic}
\usepackage{graphicx}
\usepackage{textcomp}
\usepackage{xcolor}
\usepackage{balance}
\usepackage{amsmath,dsfont,bbm,epsfig,amssymb,amsfonts,amstext,verbatim,amsopn,cite,subfigure,multirow,multicol,lipsum,xfrac}
\usepackage{amsthm}
\usepackage{mathtools,amsthm}
\usepackage{perpage}
\usepackage{balance}
\usepackage{url}
\usepackage{amsfonts}
\usepackage{epsfig}
\usepackage[font={small}]{caption}
\usepackage{psfrag}	
\usepackage{etoolbox}
\usepackage{pifont}
\usepackage[utf8]{inputenc}
\usepackage[T1]{fontenc}  
\usepackage[nolist]{acronym}
\usepackage{paralist}
\usepackage{enumitem}
\usepackage{bbm}
\usepackage{tikz}
\usepackage{pgfplots}
\usetikzlibrary{shapes,arrows}

\newcommand{\jj}{\mathrm{j}}
\newcommand{\Ex}[1]{ \mathcal{E} \left[  #1 \right]}
\newcommand{\dbc}[1]{ \left[  #1 \right]}
\newcommand{\tr}[1]{\mathrm{tr} \left\lbrace #1 \right\rbrace}

\newcommand{\brc}[1]{ \left(  #1 \right)}
\newcommand{\abs}[1]{  \left\vert  #1 \right\vert}
\newcommand{\norm}[1]{  \left\Vert  #1 \right\Vert}

\newcommand{\her}{\mathsf{H}}
\newcommand{\xx}{{\mathrm{x}}}
\newcommand{\yy}{{\mathrm{y}}}
\newcommand{\trp}{{\mathsf{T}}}

\newcommand{\setC}{{\mathbb{C}}}
\newcommand{\mR}{{\mathbf{R}}}

\newcommand{\mI}{{\mathbf{I}}}

\newtheorem{remark}{Remark}
\newtheorem{proposition}{Proposition}
\newtheorem{corollary}{Corollary}

\def\BibTeX{{\rm B\kern-.05em{\sc i\kern-.025em b}\kern-.08em
    T\kern-.1667em\lower.7ex\hbox{E}\kern-.125emX}}
\begin{document}
	
	\begin{acronym}
		\acro{mimo}[MIMO]{multiple-input multiple-output}
		\acro{miso}[MISO]{multiple-input single-output}
		\acro{siso}[SISO]{single-input single-output}
		\acro{csi}[CSI]{channel state information}
		\acro{awgn}[AWGN]{additive white Gaussian noise}
		\acro{iid}[i.i.d.]{independent and identically distributed}
		\acro{ut}[UT]{user terminal}
		\acro{bs}[BS]{base station}
		\acro{snr}[SNR]{signal-to-noise ratio}
		\acro{rf}[RF]{radio frequency}
		\acro{los}[LoS]{line of sight}
		\acro{nlos}[NLoS]{non-line of sight}
		\acro{irs}[IRS]{intelligent reflecting surface}
		\acro{ula}[ULA]{uniform linear array}
		\acro{aoa}[AoA]{angle-of-arrival}
		\acro{aod}[AoD]{angle-of-departure}
		\acro{mrt}[MRT]{maximum ratio transmission}
		\acro{noma}[NOMA]{non-orthogonal multiple access}
		\acro{mse}[MSE]{mean squared error}
	\end{acronym}

\title{How Should IRSs Scale to Harden Multi-Antenna Channels?
\thanks{This work has been accepted for presentation in 2022 IEEE Sensor Array and Multichannel Signal Processing Workshop (SAM ) in Trondheim, Norway. The link to the final version in the Proceedings of SAM will be available later.}
}

\author{
	\IEEEauthorblockN{
		Ali Bereyhi\IEEEauthorrefmark{1},
		Saba Asaad\IEEEauthorrefmark{1},
		Chongjun Ouyang\IEEEauthorrefmark{4},
		Ralf R. M\"uller\IEEEauthorrefmark{1},
		Rafael F. Schaefer\IEEEauthorrefmark{3}, and
		H. Vincent Poor\IEEEauthorrefmark{2}}
\IEEEauthorblockA{
	\IEEEauthorrefmark{1}Friedrich-Alexander Universit\"at Erlangen-Nürnberg (FAU), 
	\IEEEauthorrefmark{4}Beijing University of Posts and Telecommunications, \\
	\IEEEauthorrefmark{3}University of Siegen, 
	\IEEEauthorrefmark{2}Princeton University
	%\thanks{ Saba Asaad, Ali Bereyhi and Ralf Müller are with the Friedrich-Alexander Universit\"at; emails: \{saba.asaad, ali.bereyhi,ralf.r.mueller\}@fau.de. Chongjun Ouyang is with the Beijing University of Posts and Telecommunications; email: dragonaim@bupt.edu.cn. H. Vincent Poor is with the Princeton University; email: poor@princeton.edu.}
	\thanks{This work was supported in part by the German Research Foundation (DFG), under Grant No. MU 3735/7-1, in part by the German Federal Ministry of Education and Research (BMBF) under Grant 16KIS1242, and in part by the U.S. National Science Foundation under Grant CCF-1908308.}}}

\maketitle

\begin{abstract}
This work extends the concept of channel hardening to multi-antenna systems that are aided by intelligent reflecting surfaces (IRSs). For fading links between a multi-antenna~transmitter and a single-antenna receiver, we derive~an~accurate~approximation for the distribution of the input-output~mutual~information when the number of reflecting elements grows large. The asymptotic results demonstrate that by increasing the number of elements on the IRS, the end-to-end channel hardens as long as the physical dimensions of the IRS grow as well. The growth rate however need not to be of a specific order and can be significantly sub-linear. The validity of the analytical result is confirmed by numerical experiments.
\end{abstract}


\section{Introduction}
Channel hardening is a fundamental large-scale property of \ac{mimo} systems. This property indicates that key performance metrics\footnote{For instance, the channel capacity. Such metrics are in general random, due to the randomness of the fading process.} of a \ac{mimo} channel become deterministic, as the dimensions of the \ac{mimo} channel grow large on at least one side \cite{hochwald2004multiple,asaad2018massive}. From the information-theoretic point of view, this is the key property that leads to significant performance gains of massive \ac{mimo} systems \cite{bjornson2017massive}.

This study aims to extend analytically the notion of channel hardening to multi-antenna systems that are enhanced by \acp{irs}; see \cite{wu2019towards,li2019joint,tang2020joint,wu2021irs,bereyhi2020secure,asaad2021designing,asaad2022irs,zhang2021reconfigurable} and the references therein for discussions on \ac{irs}-aided \ac{mimo} systems and their applications. Here, we address this question: How do \ac{irs}-aided \ac{mimo} systems with \textit{finite} transmit and receive array sizes behave as the \textit{number of \ac{irs} elements} grows large?

\subsection{Related Work and Main Contributions}
Several studies address the above question considering some simplified models. In \cite{bjornson2020rayleigh}, channel hardening is discussed in the context of an \ac{irs}-aided setting with a single-antenna transmitter and receiver. Extensions to scenarios with multiple \acp{irs} and other fading models are given in \cite{zhang2019analysis,ibrahim2021exact}. The studies in \cite{wang2021performance} and \cite{jung2020performance} further investigate channel hardening for~\ac{irs}-aided \ac{noma} systems and a fully stand-alone \ac{irs}-based transmitter, respectively.

The studies most closely related to this work are \cite{wang2020intelligent,wang2021massive,zhang2021outage}. Starting from \cite{wang2020intelligent}, Wang et al. show that the conventional form~of~the channel hardening property does not hold in \ac{irs}-aided \ac{mimo} settings, i.e.,  by increasing the number of transmit antennas while keeping the number of \ac{irs} elements fixed, the end-to-end channel does not harden. The following study, i.e., \cite{wang2021massive}, demonstrates that in fact \ac{irs}-aided channels harden when the \textit{number~of~\ac{irs} elements} grows large. Finally, the study in \cite{zhang2021outage} derives the distribution of the input-output mutual information for an \ac{irs}-aided \ac{mimo} setting.

Although \cite{wang2020intelligent,wang2021massive,zhang2021outage} extend the concept of channel hardening to \ac{irs}-aided \ac{mimo} systems, their derivations are carried~out under several unrealistic assumptions, e.g., simplified channel models. We deviate from these simplifications and give~a~rigorous large-system characterization of the input-output mutual information for an \ac{irs}-aided setting. The analysis takes into account the physical scaling of the \ac{irs} and its impacts on the spatial correlation among the \ac{irs} elements. The derivations lead to this interesting finding which agrees with intuition: In a multi-antenna system being aided with an \ac{irs} whose physical dimensions grow, the end-to-end channel hardens, regardless of how fast the \ac{irs} dimensions grow. The hardening speed however depends on the growth order.

\subsection{Notation}
Scalars, vectors and matrices are indicated by nonbold,~bold lower-case, and bold upper-case letters, respectively.~The~transposed conjugate of $\mathbf{H}$ is shown by $\mathbf{H}^{\mathsf{H}}$, and $\dbc{\mathbf{H}}_{nm}$ is the entry of $\mathbf{H}$ at row $n$ and column $m$. The $N\times N$ identity and all-one matrices are denoted by $\mathbf{I}_N$ and $\boldsymbol{1}_N$, respectively. The mathematical expectation is denoted by $\mathcal{E}$ and $\mathcal{CN}\left( \eta,\sigma^2\right) $ is a complex Gaussian distribution with mean $\eta$ and variance $\sigma^2$.

\section{System Model}
\label{sec:sys}
We consider a \ac{miso} system in which a transmitter with an $M$-antenna array communicates with a single-antenna receiver. An \ac{irs} with $N$~reflecting~elements is further established to enhance the communication link between the transmitter and the receiver. Each element of the \ac{irs} reflects its received signal after applying a tunable phase-shift. We assume that both direct and reflection links are available and experience slow and frequency-flat fading processes. The received signal at the destination is hence the superposition of two components: one that is received through the direct path, and one that is reflected by the \ac{irs}.

Let $x_m$ denote the symbol sent by the $m$-th element of the transmit array satisfying the transmit power constraint 
\begin{align}
	\sum_{m=1}^{M} \Ex{\abs{x_m}^2} \leq \rho
\end{align}
for transmit power $\rho$. The received signal is then given by
\begin{align}\label{eq:rx}
	y=\sum_{m=1}^{M} h_{\mathrm{d}, m} x_m +\sum_{m=1}^{M} \sum_{n=1}^{N} \mathrm{e}^{-\jj\beta_n} h_{\mathrm{r}, n} t_{nm} x_m + z,
\end{align}
where
\begin{itemize}
	\item $z$ is zero-mean and unit-variance \ac{awgn}, 
	\item $h_{\mathrm{d}, m}$ is the direct channel coefficient between element $m$ at the transmit array and the receiver,
	\item $t_{nm}$ is the channel coefficient between the $m$-th element at the transmit array and the $n$-th reflecting element,
	\item $h_{\mathrm{r}, n}$ is the channel coefficient between the $n$-th reflecting element and the receiver, and
	\item $\beta_n$ is the phase-shift applied by reflecting element $n$.
\end{itemize}
The channel state information is assumed to be known at both sides of the channel.
\subsection{Channel Model}
The transmit array and \ac{irs} are assumed to be rectangular uniform planar arrays with $M_\xx$ and $N_\xx$ horizontal and $M_\yy$ and $N_\yy$ vertical isotropic elements, respectively, i.e., $M=M_\xx M_\yy$ and $N=N_\xx N_\yy$.~Each pair of neighboring transmit antennas are distanced with $\ell_\xx$ and $\ell_\yy$ on the horizontal and vertical axes, respectively. The horizontal and vertical distances between two neighboring reflecting elements are further denoted by $d_\xx$ and $d_\yy$, respectively. We assume that $\ell_\xx$, $\ell_\yy$, $d_\xx$ and $d_\yy$ are smaller than a half wave-length.

%\cite{bjornson2020rayleigh,zhang2021reconfigurable}

We consider a classical scenario in which the \ac{los} link in the direct path\footnote{Note that the direct path is different from the line of sight.} is blocked. %We further assume~$\ell_\xx$ and $\ell_\yy$ are large enough, such that the spatial correlation at the transmitter can be ignored\footnote{This is feasible, since we do not assume $M$ to be large.}. 
The~direct~path~between antenna $m$ at the transmitter and the receiver is hence modeled by as a standard Rayleigh fading process, i.e., 
%\begin{align}
	$h_{\mathrm{d},m}=\sqrt{\alpha_\mathrm{d} A_M } \tilde{h}_{\mathrm{d},m}$, %
%\end{align}
where $\alpha_\mathrm{d}$ models the effective path-loss, $A_M$ denotes the area of a single element on the transmit array, i.e.,  $A_M = \ell_\xx \ell_\yy$, and $\tilde{h}_{\mathrm{d},m}$ is zero-mean and unit-variance complex Gaussian, i.e., $\tilde{h}_{\mathrm{d},m}\sim \mathcal{CN}\left(0, 1\right)$. In the sequel, we compactly denote~the direct channel vector as $\mathbf{h}_\mathrm{d}=[h_{\mathrm{d},1}, \cdots, h_{\mathrm{d},M}]^\trp$.

The \ac{irs} is often deployed flexibly in the network. We hence assume that the \ac{irs} is located at a moderate distance in the transmitter sight, such that the communication link between the transmitter and the \ac{irs} is dominated by a \ac{los} component. As a result, we represent the channel from the transmitter to the \ac{irs} by ${\mathbf{T}}\in \mathbb{C}^{N \times M}$, where $t_{nm}=[{\mathbf{T}}]_{nm} = %$ denotes the channel coefficient between the $m$-th element at the transmit array and the $n$-th reflecting element and is given by
%\begin{align}
	%t_{nm}=
	\sqrt{\alpha_{\mathrm{s}} A_N }\bar{t}_{nm}$.
%\end{align}
Here, $\alpha_{\mathrm{s}}$ is the path-loss, $A_N$ is the area of a reflecting element, i.e., $A_N = d_\xx d_\yy$, and $\bar{t}_{nm}$ denotes the \ac{los} component.

\begin{remark}
	In general, the effective area of a single element depends on the wave-length. The considered simple model for $A_N$ and $A_M$ however follows from the fact that we assume the neighboring elements on the transmit array and the \ac{irs} to be distanced less than a half wave-length.
\end{remark}

In practice, the receiver is in a relatively large~distance~from both transmitter and \ac{irs}; however, it can yet be in~the~sight~of the \ac{irs}. We hence assume that the link between the \ac{irs} and the receiver has both \ac{los} and \ac{nlos} components. This means that the coefficient of the channel between the $n$-th reflecting element and the receiver is modeled as
\begin{align} \label{eq:hr} 
	h_{\mathrm{r}, n}=\sqrt{\alpha_{\mathrm{ r}} A_N }\left(\sqrt{\frac{\kappa_\mathrm{r}}{\kappa_\mathrm{r}+1}}\bar{h}_{\mathrm{r}, n} +\sqrt{\frac{1}{\kappa_\mathrm{r}+1}}\tilde{h}_{\mathrm{r}, n}\right), 
\end{align}
where $\alpha_{\mathrm{ r}}$ and $\kappa_\mathrm{r}$ are the path-loss and the Rician~factor,~respectively. The coefficient $\bar{h}_{\mathrm{r}, n}$ further denotes the \ac{los},~and~$\tilde{h}_{\mathrm{r}, n}$ models the small-scale fading process in the \ac{nlos} link.

In \ac{irs}-aided systems, the \ac{irs} is typically considered to be filled by a large number of reflecting elements. Consequently, the distance between neighboring elements is rather small, and hence the spatial correlation among the \ac{irs} elements cannot be ignored. To capture the spatial correlation, we assume that $\tilde{\mathbf{h}}_{\mathrm{r}}=[\tilde{h}_{\mathrm{r}, 1}, \cdots, \tilde{h}_{\mathrm{r}, N}]^\trp$ is a zero-mean complex~Gaussian~process with the covariance matrix $\mR\in \setC^{N\times N}$. Note that due to power normalization, we have $[\mR]_{nn} =1$ for $n\in[N]$. 

The \ac{los} components in the channel model can be written in terms of the array responses. Let $\lambda$ be the wavelength and %define the following operators:
\begin{align}
	\begin{array}{ll}
		i_M\brc{m} =\brc{m-1} \;  \mathrm{mod} \; M_\xx ,  &j_M\brc{m} = \left\lfloor \displaystyle \frac{m-1}{M_\xx} \right\rfloor \vspace*{2mm} \\
		i_N\brc{n} = \brc{n-1}  \;  \mathrm{mod} \;  N_\xx,  &j_N\brc{n} = \left\lfloor \displaystyle \frac{n-1}{N_\xx} \right\rfloor  
	\end{array}
\end{align}
where $x \;  \mathrm{mod} \;  L$ determines $x$ modulo $L$. We further let the exponent functions at azimuth angle $\varphi$ and elevation angle $\theta$ for transmit element $m$ and reflecting element $n$ be
\begin{subequations}
	\begin{align}
	\Phi_m\left( \varphi,\theta\right) &= {i_M\brc{m} \ell_\xx \cos{\theta} \sin{\varphi}  +  j_M\brc{m} \ell_\yy \sin{\theta} },\\
	\Pi_n\left( \varphi,\theta\right) &={i_N\brc{n} d_\xx \cos{\theta} \sin{\varphi} +  j_N\brc{n} d_\yy \sin{\theta} },
	\end{align}
\end{subequations}
respectively. Denoting the wavelength by $\lambda$, the transmit and \ac{irs} array responses are given at $\brc{\varphi,\theta}$ respectively by \cite{bjornson2020rayleigh}
\begin{subequations}
	\begin{align}
		\mathbf{a}_M\left( \varphi,\theta\right) &= \left[ \mathrm{e}^{   \tfrac{2\pi\jj }{\lambda} \Phi_1\left( \varphi,\theta\right)  }, \cdots, \mathrm{e}^{  \tfrac{2\pi\jj }{\lambda} \Phi_M\left( \varphi,\theta\right) } \right]^\trp,\\
		\mathbf{a}_N\left( \varphi, \theta\right) &= \left[ \mathrm{e}^{ \tfrac{2\pi\jj }{\lambda}  \Pi_1\left( \varphi,\theta\right)  }, \cdots, \mathrm{e}^{  \tfrac{2\pi\jj }{\lambda}  \Pi_N\left( \varphi,\theta\right) } \right]^\trp.
	\end{align}
\end{subequations}

Given the array responses, the \ac{los} components  are given~by
%Denote the matrix of \ac{los} channel coefficients between the transmitter and the \ac{irs} by $\bar{\mathbf{T}} \in \setC^{N\times M}$ whose entry $\brc{n,m}$ is $\bar{t}_{nm}$. Moreover, let the  be . We can write
%\begin{subequations}
	\begin{align}
		{\mathbf{T}} &= \sqrt{\alpha_{\mathrm{s}} A_N } \; \mathbf{a}_N\left( \varphi_{\mathrm{r}1},\theta_{\mathrm{r}1}\right) \mathbf{a}_M\left( \varphi_{\mathrm{t}2} , \theta_{\mathrm{t}2} \right)^\mathsf{H}, \label{a1}%\\
		%\mathbf{\bar{h}}_{\mathrm{r}} &= \mathbf{a}_N\left( \varphi_{\mathrm{t}1} , \theta_{\mathrm{t}1} \right), \label{a2}
	\end{align}
%\end{subequations}
%and $\mathbf{\bar{h}}_{\mathrm{r}} = \mathbf{a}_N\left( \varphi_{\mathrm{t}1} , \theta_{\mathrm{t}1} \right)$, where 
and $\mathbf{\bar{h}}_{\mathrm{r}}=[\bar{h}_{\mathrm{r}, 1}, \cdots, \bar{h}_{\mathrm{r}, N}]^\trp = \mathbf{a}_N\left( \varphi_{\mathrm{t}1} , \theta_{\mathrm{t}1} \right)$. %is the \ac{los} component of the channel between the \ac{irs} and the receiver. 
Here, $\brc{\varphi_{\mathrm{r}1},\theta_{\mathrm{r}1}}$ is the \ac{aoa} at the \ac{irs}, $\brc{\varphi_{\mathrm{t}1} , \theta_{\mathrm{t}1} }$ denotes the \ac{aod} from the \ac{irs}, and $\brc{\varphi_{\mathrm{t}2} , \theta_{\mathrm{t}2} }$ is the \ac{aod} from the transmitter.

\subsection{Channel Capacity}
Consider the end-to-end channel $\mathbf{h}=[h_1, \cdots, h_M]^\trp$, where 
\begin{align}
	h_m%&=h_{\mathrm{d}, m}+\sum_{n=1}^{N} \phi_n h_{\mathrm{r}, n} t_{nm}\\
	&=h_{\mathrm{d}, m}+\sum_{n=1}^{N} \mathrm{e}^{-\jj\beta_n} h_{\mathrm{r}, n} t_{nm}.
\end{align}
For a given realization of $\mathbf{h}$, the channel capacity is achieved by maximum ratio transmission\footnote{Note that this expression gives the  maximum input-output mutual information for the given channel realization and under an average power constraint. By considering the \textit{ergodic} capacity as the metric and taking into account the channel fluctuations over time, we need further to consider power allocation over-time which is optimally performed using the water-filling algorithm \cite{yu2004iterative}.} and is given by\cite{tse2005fundamentals,loyka2017capacity}
\begin{align}\label{eq:I}
	\mathcal{C} =\log_2 \left(1+ \rho \norm{\mathbf{h} }^2 \right).
\end{align}

Due to fading, the channel capacity expression is a random process whose statistics determine various~performance~metrics, e.g., ergodic capacity and outage probability. The main goal of this study is to characterize the statistics of $\mathcal{C}$, when $N$ is asymptotically large.

\section{Large-System Analysis}
\label{sec:main}
The basic limiting scenario is to consider an extreme case in which the number of reflecting elements grows large while the distances between neighboring elements are kept fixed. This is however an unrealistic assumption, as the \ac{irs} area in this case scales linearly with the number of reflecting elements. In practice, the \ac{irs} area is restricted, and hence by growing the number of \ac{irs} elements, the reflecting elements are distanced closer on the surface. This leads to a smaller effective area for \ac{irs} elements and therewith to higher spatial correlation \cite{ivrlac2010toward}.

To take the above scaling fact into account, we consider a basic scaling model for the \ac{irs} area. Namely, we assume that the area of each \ac{irs} element scales as~$A_N = {A_0}N^{-q}$ %
%\end{align}
for some constant $A_0$ and $0\leq q \leq 1$. This means that the total area of the \ac{irs} scales as $A_{\rm IRS} =  A_0 N^{1-q}$.
%\end{align}
This scaling addresses the limiting scenarios between the two extreme cases:
\begin{itemize}
	\item $q=0$ is the idealistic case in which the distances between neighboring elements are kept fixed, i.e., the area of each \ac{irs} element is fixed.
	\item $q=1$ addresses the case in which the total area of the \ac{irs} is fixed. In this case, the area of each reflecting element shrinks reverse-linearly in $N$.
\end{itemize}
%By setting $0< q < 1$, an intermediate scaling is considered in which the \ac{irs} surface grows sub-linearly by $N$.

\begin{remark}
	Note that for $q\neq 1$, the above scaling implies that by sending $N\rightarrow \infty$, the area of \ac{irs} also grows asymptotically large. One may thus conclude that the far-field model for the \ac{irs} array response is no longer valid in the large-system limit. To avoid such inconsistency, we assume that the distances among the terminals, i.e., transmitter, receiver and \ac{irs}, are bounded uniformly from below by $D_0N^{{\gamma}/{2}}$ for some $D_0$ and $\gamma > 1-q$. This assumption guarantees the validity of the far-field model through the asymptotic analyses. More details in this respect can be followed in \cite{ivrlac2010toward,ivrlac2014multiport}.
\end{remark}

Proposition~\ref{th:0} gives the large-system statistics of the capacity term for an arbitrary sub-linear scaling of the \ac{irs} surface.

\begin{proposition}
	\label{th:0}
	Let the area of the \ac{irs} scale sub-linearly with $N$, i.e., $A_{\rm IRS} =  A_0 N^{1-q}$ for some fixed $A_0$ and $0\leq q <1$. Let the phase-shifts be set to %the vector $ \bbeta^\star = [\beta^\star_1,\ldots,\beta^\star_N]^\trp$ with
	\begin{align} 
		\beta_n^\star= \frac{2\pi }{\lambda} \brc{\Pi_n\left( \varphi_{\mathrm{r}1} ,\theta_{\mathrm{r}1} \right)  + \Pi_n\left( \varphi_{\mathrm{t}1} ,\theta_{\mathrm{t}1} \right)}.  \label{eq:beta_n_star}
	\end{align}
	Assume the maximum eigenvalue of the \ac{irs} covariance matrix $\mR$, denoted by $\lambda_{\max}$, grows with $N$ sub-linearly, i.e.,
	\begin{align}
		\lim_{N\rightarrow\infty} {\lambda_{\max}}{N}^{-1} = 0, \label{eq:lamMax}
	\end{align}
	and satisfies
	\begin{align}
		\lim_{N\rightarrow\infty} \brc{\lambda_{\max} A_{\rm IRS}}^{-1} = 0. \label{eq:const_IRS}
	\end{align}
	Then, for large $N$, $\mathcal{C}$ is well-approximated by a real Gaussian random variable whose mean and standard deviation are %given by
	\begin{subequations}
		\begin{align}
			\mu_{\mathcal{C}} &=   \log_2 \left(1 +\rho M \mu \right), \\
			\sigma_{\mathcal{C}} &= \frac{ \displaystyle \rho M  \log_2 \mathrm{e} }{1+ \rho M \mu }  \sqrt{ \omega \eta + \eta + \frac{M-1}{M} \alpha_{\mathrm{ d}} A_M },
		\end{align}
	\end{subequations}
	respectively, for
	\begin{subequations}
		\begin{align}
			\mu &=  \alpha_{\mathrm{ d}} A_M+ \kappa_{\mathrm{ r} } \bar{\alpha}_N N^2 + \bar{\alpha}_N \mathbf{\bar{h}}_{\mathrm{r}}^\her \mR \mathbf{\bar{h}}_{\mathrm{r}} ,\\
			\eta &=  \frac{\alpha_{\mathrm{ d}}A_M}{M} +  \bar{\alpha}_N \mathbf{\bar{h}}_{\mathrm{r}}^\her \mR \mathbf{\bar{h}}_{\mathrm{r}} ,\\
			\omega &=  2\kappa_{\mathrm{ r} } \bar{\alpha}_N N^2 + \bar{\alpha}_N \mathbf{\bar{h}}_{\mathrm{r}}^\her \mR \mathbf{\bar{h}}_{\mathrm{r}} ,
		\end{align}
	\end{subequations}
	and $\bar{\alpha}_N$ being defined as $\bar{\alpha}_N= {\alpha_{\mathrm{ r}} \alpha_{\mathrm{ s}} } A_N^2 /\brc{1+\kappa_\mathrm{r}}$.
	%\end{align}
\end{proposition}

\begin{proof}
The proof is given in three steps: First, the distribution of the end-to-end \ac{snr} is derived. The mean and variance are then bounded using bounds on the Rayleigh quotient of $\mR$. By sending $N\rightarrow\infty$, the large-system approximation is derived. Due to page limit, we skip the details and refer the reader to the extended version \cite{bereyhi2022hardening}.
\end{proof}

%Proposition~\ref{th:0} characterizes asymptotically the random variable $\mathcal{C}$, and hence it determines various capacity metrics~of~the channel. As it is shown in Section~\ref{sec:Numerics}, the asymptotic~result~approximates tightly $\mathcal{C}$, even for small choices of $N$.

Proposition~\ref{th:0} illustrates how \ac{irs}-aided multi-antenna channels harden. It also specifies the required scaling~order~of~the \ac{irs} to guarantee channel hardening. % %We address these aspects in the forthcoming section. However, before we start, 
It is further worth noting few remarks regarding the constraints on $\lambda_{\max}$: %in Proposition~\ref{th:0}:
\begin{enumerate}
	\item Proposition~\ref{th:0} assumes a sub-linearly growing $\lambda_{\max}$. This can be interpreted as follows: $\lambda_{\max}$ is uniformly bounded from above as $\lambda_{\max} \leq aN^u$ for some real $a$ and $0\leq u <1$. This is not a strong constraint, as 
	%\begin{align}
		$\lambda_{\max} \leq \tr{\mR} = N$.
	%\end{align}
	\item In general, the growth order of $\lambda_{\max}$ is mutually coupled with $q$; see \cite[Section~III-C]{bereyhi2022hardening}. This is seen by considering the two extreme cases of $q$. For $q=0$, $\mR=\mI_N$~is~feasible, and hence $\lambda_{\max} = 1$, i.e., the uniform upper-bound~is valid for $a =1$ and  $u =0$. As $q \rightarrow 1$, $\mR$ tends to a rank-one matrix, and hence $\lambda_{\max} = N$, i.e., $a = u =1$ in the upper-bound of $\lambda_{\max}$. For less rank-deficient covariance matrices, the bound is given for some $u\in \dbc{0,1}$. %Similar to the constraint on the physical dimension of the \ac{irs}, we show later that despite assuming $0\leq u <1$, Proposition~\ref{th:0} is still valid for the extreme case of $u=1$. 
	\item  Considering the scaling of the \ac{irs} area and $\lambda_{\max}$, i.e., $q$ and $u$, the constraint in \eqref{eq:const_IRS} restricts $u$ in the upper-bound of $\lambda_{\max}$ to satisfy $u\geq q$. In the extended version~of the work \cite[Section~III-C]{bereyhi2022hardening}, it is further shown that this is the case for an \ac{irs} covariance matrix which is derived for the standard Rayleigh fading model in \cite{bjornson2020rayleigh}.
\end{enumerate}

\section{Asymptotic Channel Hardening}
The classical channel hardening result, i.e., the initial study in \cite{hochwald2004multiple}, indicates that $\mathcal{C}$ converges to a deterministic variable,~as the transmit array size $M$ grows unboundedly large. We are however interested in a different asymptotic regime; namely, a scenario with \textit{finite} transmit antennas but an unboundedly large number of reflecting elements, i.e., fixed $M$ and $N\rightarrow\infty$. %Our main result shows that in this alternative asymptotic regime, the channel still hardens, under a very mild constraint on the correlation among the \ac{irs} elements.

\begin{corollary}
	\label{cor:1}
	Let the \ac{irs} area scale as $A_{\rm IRS} =  A_0 N^{1-q}$ for some fixed $A_0$. Assume that the phase-shifts are set to \eqref{eq:beta_n_star}, and let $\lambda_{\max}$ be uniformly bounded from above as $\lambda_{\max} \leq a N^u$, for some $0\leq q < u <1$. As $N\rightarrow\infty$, the mean of $\mathcal{C}$ grows large and its variance converges to zero.
\end{corollary}
\begin{proof}
The proof follows from \cite[Proposition~2]{bereyhi2022hardening}, where it is shown that under the given constraints, there exist real scalars $b$ and $c$, such that $\mu_{\mathcal{C}} \geq b + \brc{1-q}  \log_2 N$ and $\sigma_{\mathcal{C}}^2 \leq {{c}{N^{u-1}}}$. Considering $0\leq q < u <1$, we can conclude that $\mu_{\mathcal{C}}$ grows large, and $\sigma_{\mathcal{C}}^2$ converges to zero as $N\rightarrow \infty$. 
\end{proof}

In the extended version \cite[Section~III-C]{bereyhi2022hardening}, it is shown for the Rayleigh fading model that the constraint $0\leq q < u \leq1$ is valid with $u=1$ if $q=1$. Noting that $q=1$ represents the case with fixed \ac{irs} area, Corollary~\ref{cor:1} indicates that the \ac{irs}-aided channel hardens, if the physical dimensions of the \ac{irs} grow large with $N$. The growth is however sufficient to be sub-linear. From an implementational viewpoint, it is a valid constraint. In fact, due to the restricted physical dimensions of each reflecting element, the distance between two neighboring elements on the surface cannot be set below a certain limit, and hence the overall area of \ac{irs} always increases in $N$.

\section{Numerical Experiments}
\label{sec:Numerics}
We now confirm the accuracy of the derivations for practical system dimensions through numerical experiments. We consider a basic scenario in which the transmitter is equipped with a $2 \times 2$ planar array and the \ac{irs} contains $N=256$ reflecting elements. The elements on the \ac{irs} are assumed to be aligned on a rectangle with $N_\xx = 8$ horizontal elements and $N_\yy = 32$ vertical elements. The elements at both transmitter and receiver are distanced with $\ell_\xx = d_\xx = \ell_\yy = d_\yy = \lambda /2$, where $\lambda$ denotes the wave-length. We further set $\alpha_{\mathrm{ d}} A_M = \alpha_{\mathrm{ r}} A_N= \alpha_{\mathrm{ s}} A_N=1$ and $\log \kappa_{\mathrm{ r} } = 0$ dB. To generate the covariance matrix, we invoke \cite[Proposition~1]{bjornson2020rayleigh} and set entry $\brc{n,n'}$ of $\mR$ to be
\begin{align}
	\dbc{ \mR }_{n n'} = \mathrm{sinc} \brc{\frac{2}{\lambda} \sqrt{d_\xx^2 \Delta_\xx^2 + d_\yy^2 \Delta_\yy^2 } } \label{eq:RR}
\end{align}
with %$D_\xx$ and $D_\yy$ being
%\begin{subequations}
	%\begin{align}
		$\Delta_\xx =  {i_N\brc{n} - i_N\brc{n'}}$  and %\\
		$\Delta_\yy =  {j_N\brc{n} - j_N\brc{n'}}$.
	%\end{align}
%\end{subequations}
The \ac{aoa} and \acp{aod} are further set to $\brc{\varphi_{\mathrm{r}1},\theta_{\mathrm{r}1}} = \brc{\pi/6,\pi/3 }$, $\brc{\varphi_{\mathrm{t}1},\theta_{\mathrm{t}1}} = \brc{\pi/8,2\pi/3 }$ and $\brc{\varphi_{\mathrm{t}2},\theta_{\mathrm{t}2}} = \brc{\pi/7,\pi/5 }$. The power constraint is set to $\rho=1$. 

We collect $10^5$ realizations of the channel and determine~$\mathcal{C}$ for each realization.  The empirical density is then determined from the collected data and compared with Proposition~\ref{th:0} in Fig.~\ref{fig:dist}. As the figure shows, the analytical result of Proposition~\ref{th:0} almost perfectly matches the empirical density. 
\begin{figure}
	\centering
	\input{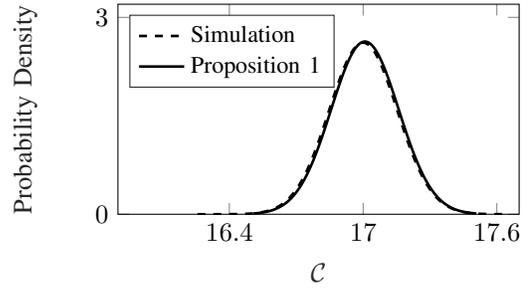}
	\caption{Density of $\mathcal{C}$ and the empirical density fitted to the data.}
	\label{fig:dist}
\end{figure}

As the next experiment, we replace the rectangular \ac{irs} in the above setting with a square array, i.e., $N_\xx = N_\yy$ and let $N$ grow gradually from $N= 64$ to $N=1296$ assuming that the distance between each two neighboring elements remains $\lambda/2$. The variance of $\mathcal{C}$ is then plotted against $N$ using both numerical data and asymptotic expression~in~Proposition~\ref{th:0}.~As Fig.~\ref{fig:upp} demonstrates, the analytical results closely approximate numerical simulations, even for rather small choices of $N$. The figure further depicts the drop of $\sigma^2_{\mathcal{C}}$ against $N$ indicating the hardening of the end-to-end channel. %with the growth of \ac{irs} size.

\begin{figure}
	\centering
	\input{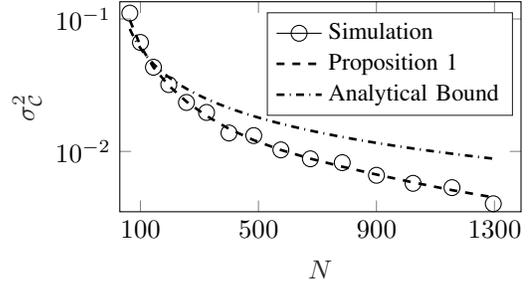}
	\caption{Asymptotic channel hardening with respect to the IRS size.}
	\label{fig:upp}
\end{figure}

Fig.~\ref{fig:upp} further compares $\sigma^2_{\mathcal{C}}$ to the analytical upper-bound on the variance of $\mathcal{C}$ derived in the extended version of the work \cite[Proposition~2]{bereyhi2022hardening}. Interestingly, the suggested upper bound gives a \textit{pessimistic} approximation of the hardening speed. In fact, the true variance drops~much faster than the upper bound. This observation suggests that the end-to-end channel hardens rather fast, even with strongly correlated reflecting elements. The consistency of this conjecture is demonstrated via several numerical investigations in the extended version \cite{bereyhi2022hardening}. Further discussions on this respect are skipped due to page limitation. The interested reader is referred to the extended version \cite{bereyhi2022hardening}.

\section{Conclusions}
By increasing the number of reflecting elements, \ac{irs}-aided channels harden as long as the physical dimensions of the \ac{irs} grow as well. However, the growth order can be significantly sub-linear. This is realistic from the implementational viewpoint, as neighboring elements on an \ac{irs} cannot be set closer than a threshold distance, due to their physical dimensions.

The above result indicates that in \ac{irs}-aided systems, even by packing the reflecting elements compactly~on~the \ac{irs}, the end-to-end channel hardens as the number of elements grows large. This finding shows that by enhancing a multi-antenna system of feasible dimensions via large \acp{irs}, the large-scale properties of a multi-antenna system can be obtained.

The results of this study can be used to investigate \ac{irs}-aided multi-antenna systems in various respects. Some examples are addressed in the extended version of this work \cite{bereyhi2022hardening}.

\bibliographystyle{IEEEtran}
\balance
\bibliography{IEEEabrv,references}
\end{document}